\documentclass
{svjour3}                     
\smartqed  

\usepackage{graphicx} 
\usepackage{times} 
\usepackage{amsmath} 
\usepackage{amssymb}  

\def\<{\leqslant}           
\def\>{\geqslant}           

\def\d{\partial}
\def\wh{\widehat}
\def\wt{\widetilde}

\def\Re{{\rm Re}}   
\def\Im{{\rm Im}}   

\def\col{{\rm vec}}   
\def\cH{{\cal H}}   
\def\mA{{\mathbb A}}    
\def\mR{{\mathbb R}}    

\def\Tr{{\rm Tr}}       
\def\rT{{\rm T}}        
\def\diam{\diamond}       
\def\bS{{\mathbf S}}

\def\bE{{\bf E}}    

\def\[[[{[\![\![}
\def\]]]{]\!]\!]}

\def\bra{{\langle}}
\def\ket{{\rangle}}

\def\Bra{\left\langle}
\def\Ket{\right\rangle}

\def\re{{\rm e}}        
\def\rd{{\rm d}}        

\def\fM{{\mathfrak M}}
\def\fN{{\mathfrak N}}
\def\fO{{\mathfrak O}}

\def\cL{{\mathcal L}}

\def\bA{{\mathbf A}}

\def\bJ{{\bf J}}

\def\br{{\bf r}}
\def\x{\times}
\def\ox{\otimes}

\def\cZ{{\mathcal Z}}

\def\bH{{\mathbf H}}

\def\cW{{\mathcal W}}

\def\cX{{\mathcal X}}

\def\cC{{\cal C}}

\def\cA{{\cal A}}
\def\cB{{\cal B}}
\def\cE{{\mathcal E}}

\def\cT{{\mathcal T}}

\def\mU{{\mathbb U}}

\def\mS{{\mathbb S}}

\def\veps{\varepsilon}

\begin{document}
\title{A dynamic programming
approach to finite-horizon coherent quantum LQG control\thanks{The work is supported by the Australian Research Council. A shortened version of this paper is to appear in the Proceedings of the Australian Control Conference, 10--11 November 2011, Melbourne, Australia. }
}


\author{Igor G. Vladimirov \and
        Ian R. Petersen 
}


\institute{I.G.Vladimirov ($\boxtimes$)\at
              School of Engineering and Information Technology, University of New South Wales at the Australian Defence Force Academy, Canberra ACT 2600, Australia\\
              \email{igor.g.vladimirov@gmail.com}           
           \and
           I.R.Petersen \at
              School of Engineering and Information Technology, University of New South Wales at the Australian Defence Force Academy, Canberra ACT 2600, Australia              \\
              \email{i.r.petersen@gmail.com}           
}

\date{}

\maketitle

\begin{abstract}
The paper is concerned with the coherent quantum Linear Quadratic Gaussian (CQLQG) control problem for time-varying quantum plants governed by linear quantum stochastic differential equations over a bounded time interval. A controller is sought among quantum linear systems satisfying physical realizability (PR) conditions.  The latter describe the  dynamic equivalence of the system to an open quantum harmonic oscillator  and relate its state-space  matrices to the free Hamiltonian, coupling and scattering operators of the  oscillator.
Using the Hamiltonian parameterization of PR controllers, the CQLQG problem is recast into an optimal control problem for a deterministic    system governed by a differential Lyapunov equation. The state of this subsidiary system is the symmetric part of the quantum covariance matrix of the plant-controller state vector. The resulting covariance control problem is  treated using  dynamic programming and Pontryagin's minimum principle. The associated Hamilton-Jacobi-Bellman equation  for the minimum cost function involves  Frechet differentiation with respect to matrix-valued variables. The gain matrices of the CQLQG optimal  controller  are shown to satisfy a quasi-separation property as a weaker quantum counterpart of the filtering/control decomposition of classical LQG controllers.
%

\end{abstract}

\section{Introduction}

Quantum feedback control, which deals with dynamical systems whose variables are noncommutative linear operators on a Hilbert space  governed by the laws of quantum mechanics (see, for example, \cite{WM_2009}),  involves two major paradigms. One of them employs classical information on the quantum mechanical system retrieved through a macroscopic measuring device and thus accompanied by decoherence and the loss of quantum information, as reflected, for example, in the projection postulate \cite{H_2001}.
The other, less ``invasive'',  approach, apparently practiced by nature to stabilize matter on an atomic scale,  is through direct interaction of quantum mechanical systems, possibly mediated by light fields. With the advances in quantum optics and nanotechnology making it possible to manipulate such interconnection,  measurement-free coherent quantum controllers provide an important alternative to the classical observation-actuation control loop.

Coherent quantum feedback can be  implemented, for example,  using quantum-optical components, such as  optical cavities, beam splitters, phase shifters, and modelled by linear quantum stochastic differential equations (QSDEs) \cite{HP_1984,M_1995,P_1992} corresponding to open quantum harmonic oscillators \cite{EB_2005,GZ_2004}. The associated notion of physical realizability  (PR) \cite{JNP_2008,P_2010,SP_2009}  reflects the dynamic equivalence of a system to an open quantum harmonic oscillator.  Being organized as quadratic constraints on  the state-space matrices,  the PR conditions imposed on the controller  complicate the solution of quantum analogues to the classical Linear Quadratic Gaussian (LQG) and $\cH_{\infty}$-control problems, and it is particularly so in regard to the Coherent Quantum LQG (CQLQG) problem \cite{NJP_2009} which has yet to be solved.

The CQLQG control problem seeks a PR quantum controller to minimize the average output ``energy'' of the closed-loop system, described by the quantum expectation of a quadratic form of its state variables. For the original infinite-horizon time-invariant setting of the problem, a numerical procedure was proposed in \cite{NJP_2009} to compute \textit{suboptimal} controllers, and algebraic equations were obtained in \cite{VP_2011} for the \textit{optimal} CQLQG controller. Subtle coupling of the equations, which comes from the PR constraints,  is apparently related to the complicated nature of the suboptimal control design procedure. This suggests that an  alternative viewpoint needs to be taken for a better understanding of the structure of the optimal quantum controller.

In the present paper, the CQLQG control problem is approached by considering its time-varying version,  with a PR quantum controller being sought to minimize the average output energy of the closed-loop system over a bounded time interval.
We outline a dynamic programming approach to the finite-horizon time-varying CQLQG problem by recasting it as a deterministic  optimal control problem for a dynamical system governed  by a differential Lyapunov equation. The state of the subsidiary system is the symmetric part of the quantum covariance matrix of the plant-controller state vector. The role of control in this covariance control \cite{SIG_1998} problem is played by a triple of matrices from the Hamiltonian parameterization of a PR controller which relates its state-space matrices to the free Hamiltonian, coupling and scattering operators of an open quantum harmonic oscillator \cite{EB_2005}.

The dynamic programming  approach to the covariance control problem is developed in conjunction with Pontryagin's minimum  principle \cite{PBGM_1962}. The appropriate costate of the subsidiary dynamical system is shown to coincide with the observability Gramian of the underlying closed-loop system.
The resulting Hamilton-Jacobi-Bellman equation (HJBE) for the minimum cost function of the closed-loop system state covariance matrix involves  Frechet differentiation in noncommutative matrix-valued variables. Such partial differential equations (PDEs) were considered, for example,  in \cite{VP_2010a} in a different context of entropy variational problems for Gaussian diffusion processes.

Using the invariance of PR quantum controllers under the group of symplectic equivalence transformations of the state-space matrices (which is a salient feature of such controllers), we establish \textit{symplectic invariance} of the minimum cost function.      This reduces the minimization of Pontryagin's control Hamiltonian \cite{SW_1997} to two independent quadratic optimization problems which yield the gain matrices of the optimal CQLQG controller.  As in the time-invariant case \cite{VP_2011}, this partial decoupling is a weaker quantum counterpart of the filtering/control separation principle of classical LQG controllers  \cite{KS_1972}. The equations for the optimal quantum controller 
involve the inverse of special self-adjoint operators on matrices \cite{VP_2011}, which can be carried out through the  matrix vectorization  \cite{M_1988,SIG_1998}.

The paper is organised as follows.  Sections~\ref{sec:plant} and \ref{sec:controller}  specify the quantum plants and coherent quantum controllers being  considered. Section \ref{sec:PR} revisits PR conditions to make the exposition  self-contained. Section~\ref{sec:problem} formulates the CQLQG control problem. Section~\ref{sec:symplectic} derives a PDE for the minimum cost function from the symplectic invariance of PR controllers. This PDE is solved in Appendix \ref{sec:Vshape_proof} to provide an insight into the structure of the function.  Section \ref{sec:HJBE} establishes the HJBE for the minimum cost function and identifies the costate in the related Pontryagin's  minimum  principle as the observability Gramian. Section \ref{sec:gain} carries out the minimization involved in the HJBE and obtains equations for the optimal controller gain matrices, using the special  linear operators from Appendix \ref{sec:class}. Section \ref{sec:eqsummary} summarizes the system of equations for the optimal CQLQG controller. Section \ref{sec:conclusion} provides concluding remarks and outlines further research.

\section{Open quantum plant}\label{sec:plant}

The quantum plant considered below is an open quantum system which is coupled to another such system (playing  the role of a controller), with the dynamics of both systems affected by the environment.
At any time $t$, the plant is described by a $n$-dimensional vector $x_t$  of self-adjoint operators 
on a Hilbert space, with $n$ even.  The plant state vector $x_t$ evolves  in time and contributes to a $p_1$-dimensional output of the plant  $y_t$ (also with self-adjoint operator-valued entries) according to QSDEs
\begin{align}
\label{x}
    \rd x_t
    & =
    A_t x_t\rd t  +  B_t \rd w_t + E_t \rd \eta_t,\\
\label{y}
    \rd y_t
    & =
    C_t x_t\rd t  +  D_t \rd w_t.
\end{align}
Here, the matrices
$
    A_t\in \mR^{n\x n}
$,
$
    B_t\in \mR^{n\x m_1}
$,
$
    C_t\in \mR^{p_1\x n}
$,
$
    D_t\in \mR^{p_1\x m_1}
$,
$
    E_t\in \mR^{n\x p_2}
$ are known deterministic functions of time, which are assumed to be continuous for well-posedness of the QSDEs,
\begin{equation}
\label{z}
    z_t
    :=
     C_tx_t
\end{equation}
is the ``signal part'' of the plant output $y_t$, and $\eta_t$ is the output of the controller to be described in Section~\ref{sec:controller}.  The noise from the environment is represented by
an $m_1$-dimensional quantum Wiener process $w_t$  (with $m_1$ even) on the boson Fock space  \cite{P_1992} with a canonical Ito table
\begin{equation}
\label{wtable}
        \rd w_t \rd w_t^{\rT}
    =
    (I_{m_1}+iJ_1/2)\rd t.
\end{equation}
Here, $i$ is the imaginary unit, $I_m$ is the identity matrix of order $m$ (with the subscript sometimes omitted), and $J_1$ is a real antisymmetric matrix, which is given by
\begin{equation}
\label{J1}
    J_1
    :=
    I_{\mu_1}
    \ox
    \bJ,
    \qquad
    \bJ
    :=
    {\begin{bmatrix}
    0 & 1\\
    -1 & 0
    \end{bmatrix}}
\end{equation}
(with $\ox$ the Kronecker product of matrices, and $\mu_1:= m_1/2$)
and specifies the  canonical commutation relations (CCRs) for the quantum noise of the plant as
\begin{equation}
\label{wCCR}
        [
            \rd w_t, \rd w_t^{\rT}
        ]
    :=
        \rd w_t \rd w_t^{\rT}
        -
        (\rd w_t \rd w_t^{\rT})^{\rT}
    =
    i J_1\rd t.
\end{equation}
Vectors are assumed to be organized as columns unless indicated otherwise, and the transpose $(\cdot)^{\rT}$ acts on vectors and
matrices with operator-valued entries as if the latter were scalars.
Accordingly, the $(j,k)$th entry of the matrix  $[W,W^{\rT}]$, associated with a vector $W$ of operators $W_1, \ldots, W_r$,  is the commutator
$$
    [W_j,W_k]
    :=
    W_jW_k-W_kW_j.
$$
Also, $(\cdot)^{\dagger}:= ((\cdot)^{\#})^{\rT}$ denotes the transpose
of the entry-wise adjoint $(\cdot)^{\#}$. In application to ordinary matrices, $(\cdot)^{\dagger}$ is the complex conjugate transpose $(\overline{(\cdot)})^{\rT}$ and will be written as $(\cdot)^*$.

\section{Coherent quantum controller}\label{sec:controller}

A measurement-free coherent quantum controller is another quantum system with a $n$-dimensional state vector $\xi_t$ with self-adjoint operator-valued entries whose interconnection with the plant (\ref{x})--(\ref{z}) is described by the QSDEs
\begin{align}
\label{xi}
    \rd \xi_t
     & =
    a_t\xi_t\rd t + b_t \rd \omega_t + e_t\rd y_t,\\
\label{eta}
    \rd \eta_t
     & =
    c_t\xi_t \rd t + d_t\rd \omega_t.
\end{align}
Here,
$
    a_t \in \mR^{n\x n}
$,
$
    b_t\in \mR^{n\x m_2}
$,
$
    c_t\in \mR^{p_2\x n}
$,
$
    d_t\in \mR^{p_2\x m_2}
$,
$
    e_t\in \mR^{n\x p_1}
$
are deterministic  continuous functions of time, and,
similarly to (\ref{z}), the process
\begin{equation}
\label{zeta}
    \zeta_t
     :=
    c_t\xi_t
\end{equation}
is the signal part of the controller output $\eta_t$. The process $\omega_t$ in (\ref{xi}) and (\ref{eta}) is the  controller noise
which is assumed to be an $m_2$-dimensional quantum Wiener process (with $m_2$ even) that
commutes with the plant noise $w_t$ in (\ref{x}) and (\ref{y}) and also has a canonical Ito table
$        \rd \omega_t \rd \omega_t^{\rT}
    =
    (I_{m_2}+iJ_2/2)\rd t
$
with the CCR matrix
\begin{equation}
\label{J2}
    J_2
    :=
    I_{\mu_2}
    \ox
    \bJ,
\end{equation}
where $\mu_2:= m_2/2$. In view of (\ref{xi}), the matrices $b_t$ and $e_t$ will be referred to as the controller noise and observation gain matrices, even though $y_t$ is not an observation signal in the classical control theoretic sense.
The combined set of equations
 (\ref{x})--(\ref{z}) and (\ref{xi})--(\ref{zeta})
describes the fully quantum closed-loop system shown in Fig.~\ref{fig:system}.
\begin{figure}[htpb]
\centering
\unitlength=1mm
\begin{picture}(50.00,24.00)
    \put(5,5){\dashbox(40,20)[cc]{}}
    \put(10,10){\framebox(10,10)[cc]{plant}}
    \put(30,10){\framebox(10,10)[cc]{contr.}}
    \put(0,15){\vector(1,0){10}}
    \put(50,15){\vector(-1,0){10}}
    \put(30,18){\vector(-1,0){10}}
    \put(20,12){\vector(1,0){10}}
    \put(-2,15){\makebox(0,0)[rc]{$w_t$}}
    \put(25,19){\makebox(0,0)[cb]{$\eta_t$}}
    \put(52,15){\makebox(0,0)[lc]{$\omega_t$}}
    \put(25,10){\makebox(0,0)[ct]{$y_t$}}
\end{picture}
\caption{
    The plant and controller form a closed-loop quantum system described by  (\ref{x})--(\ref{z}) and (\ref{xi})--(\ref{zeta}), which is influenced by the environment through the quantum Wiener processes  $w_t$ and $\omega_t$.
}
\label{fig:system}
\end{figure}
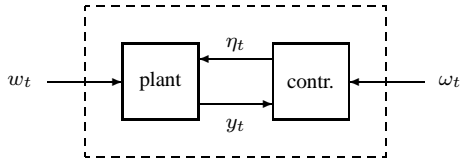
The process $\zeta_t$ in (\ref{zeta}) is analogous to the actuator signal in  classical control theory. Following the  classical approach, the performance of the quantum controller is described in terms of an $r$-dimensional process
\begin{equation}
\label{cZ}
    \cZ_t
     =
    F_t x_t  + G_t \zeta_t.
\end{equation}
Its entries are linear combinations of the plant state and ``actuator output'' variables whose relative importance  is specified by the matrices $F_t\in \mR^{r\x n}$ and $G_t\in \mR^{r\x m_2}$ which are known continuous deterministic functions of time $t$. The weighting matrices $F_t$, $G_t$ are free from physical constraints and their choice is dictated by the control design preferences. The $2n$-dimensional combined state vector
$$
    \cX_t
    :=
    \begin{bmatrix}
        x_t\\
        \xi_t
    \end{bmatrix}
$$
of the
closed-loop system  and the ``output'' $\cZ_t$ in (\ref{cZ}) are governed by the QSDEs
\begin{equation}
\label{closed}
    \rd \cX_t
      =
      \cA_t     \cX_t\rd t +   \cB_t      \rd \cW_t,
    \qquad
    \cZ_t
     =
      \cC_t       \cX_t,
\end{equation}
driven by the combined $(m_1+m_2)$-dimensional quantum Wiener process
$$
    \cW_t:=
    \begin{bmatrix}
        w_t\\
        \omega_t
    \end{bmatrix}
$$
with a block diagonal Ito table
\begin{equation}
\label{cWtable}
        \rd \cW_t \rd \cW_t^{\rT}
    =
    (I_{m_1+m_2}+iJ/2)\rd t,
    \qquad
    J
    :=
    I_{\mu_1 + \mu_2}
    \ox \bJ
\end{equation}
in conformance with (\ref{wtable})--(\ref{wCCR}), (\ref{J2}).
The state-space matrices  of the closed-loop system (\ref{closed}) are given by
\begin{equation}
\label{cABC}
    \left[
        \begin{array}{c|c}
              \cA_t     &   \cB_t     \\
            \hline
              \cC_t       & 0
        \end{array}
    \right]
    =
    {\left[
        \begin{array}{cc|cc}
                  A_t       & E_t c_t         & B_t       & E_td_t\\
                  e_t C_t   & a_t         & e_t D_t   & b_t\\
            \hline
                  F_t & G_t c_t &  0 & 0
        \end{array}
    \right]}.
\end{equation}
\section{Physical realizability conditions}\label{sec:PR}

The CCRs for the closed-loop system state vector are described by a real antisymmetric matrix
\begin{equation}
\label{cXCCR}
    \Theta_t
    :=
    -i[\cX_t, \cX_t^{\rT}]
    =
    2\Im \bE(\cX_t \cX_t^{\rT}),
\end{equation}
which, up to a factor of 2,  coincides with the entrywise imaginary part of the quantum covariance matrix of $\cX_t$, with $\bE(\cdot)$ the quantum expectation (associated, in what follows,  with the vacuum state). The matrix $\Theta_t$ evolves in time according to a  differential Lyapunov equation
\begin{equation}
\label{Thetadot}
    \dot{\Theta}_t
    =
    \cA_t \Theta_t + \Theta_t \cA_t^{\rT}
    +
    \cB_t J \cB_t^{\rT},
\end{equation}
where $J$ is the CCR matrix of the combined quantum Wiener process $\cW_t$ from (\ref{cWtable}) in the sense that $[\rd\cW_t,\rd \cW_t^{\rT}] = iJ\rd t$.
Indeed, by employing  the ideas of \cite[Proof of Theorem~2.1 on pp. 1798--1799]{JNP_2008} and combining the quantum Ito formula with the bilinearity of the commutator as
$$
    \rd [X,Y]
    =
    [\rd X, Y] + [X, \rd Y] + [\rd X,\rd Y],
$$
and using the adaptedness of the state process $\cX_t$ and the quantum Ito product rules  \cite{P_1992}, it follows that
\begin{align*}
    \rd [\cX_t, \cX_t^{\rT}]
     = &
    [
        \cA_t \cX_t \rd t + \cB_t \rd \cW_t,
        \cX_t^{\rT}
    ] +
    [
        \cX_t,
        \cX_t^{\rT}\cA_t^{\rT} \rd t + \rd \cW_t^{\rT}\cB_t^{\rT}
    ]    \\
     & +
    [
        \cA_t \cX_t \rd t + \cB_t \rd \cW_t,
        \cX_t^{\rT}\cA_t^{\rT} \rd t + \rd \cW_t^{\rT}\cB_t^{\rT}
    ]\\
     = &
    (\cA_t
    [
        \cX_t,
        \cX_t^{\rT}
    ]
    +
    [
        \cX_t,
        \cX_t^{\rT}
    ]
    \cA_t^{\rT})\rd t
    +
    \cB_t
    [
        \rd \cW_t,
        \rd \cW_t^{\rT} ]\cB_t^{\rT},
\end{align*}
which, upon division by $i$, yields (\ref{Thetadot}). Now, suppose the initial plant and controller state vectors commute with each other, so that $[x_0, \xi_0^{\rT}] = 0$ and the matrix (\ref{cXCCR}) is initialized by
\begin{equation}
\label{Theta0}
    \Theta_0
    =
    {\begin{bmatrix}
        K_1 & 0 \\
        0 & K_2
    \end{bmatrix}},
\end{equation}
where $K_1$, $K_2$ are nonsingular real antisymmetric matrices.
Then the CCR matrix of $\cX_t$ is preserved in time if and only if
\begin{equation}
\label{Thetadot0}
    \cA_t \Theta_0 + \Theta_0 \cA_t^{\rT}
    +
    \cB_t J \cB_t^{\rT} = 0
\end{equation}
for any $t\> 0$. The left-hand side of (\ref{Thetadot0}) is always an antisymmetric matrix.
Hence, by computing two diagonal and one off-diagonal blocks of this matrix with the aid of (\ref{cABC}),  it follows that the CCR preservation is equivalent to three equations
\begin{align}
\label{CCR11}
    A_t K_1 + K_1 A_t^{\rT} + B_t J_1 B_t^{\rT} + E_t d_t J_2 d_t^{\rT}E_t^{\rT} &= 0,\\
\label{CCR12}
    E_t\underbrace{(c_t K_2 + d_t J_2 b_t^{\rT})}_{\rm controller}
    +
    \underbrace{(K_1 C_t^{\rT} + B_t J_1 D_t^{\rT})}_{\rm plant}e_t^{\rT}& =  0, \\
\label{CCR22}
    a_t K_2 + K_2 a_t^{\rT} + e_t D_t J_1 D_t^{\rT}e_t^{\rT} + b_t J_2 b_t^{\rT}  &= 0
\end{align}
to be satisfied at any time $t$.
Therefore, the fulfillment  of the equalities
\begin{align}
\label{CCR12_cont}
    c_t K_2 + d_t J_2 b_t^{\rT}& =  0, \\
\label{CCR12_plant}
    C_t K_1    + D_t J_1 B_t^{\rT}
    & =  0,
\end{align}
which pertain  to the controller and the plant, respectively, is sufficient for (\ref{CCR12}). Note that (\ref{CCR22}) and (\ref{CCR12_cont}) are the conditions of physical realizability (PR) \cite{JNP_2008,NJP_2009}  of the quantum controller in the sense of equivalence of its input-output operator to an open quantum harmonic oscillator  \cite{GZ_2004}. In a similar fashion, (\ref{CCR11}) and (\ref{CCR12_plant}) are the PR conditions for the quantum plant.


\begin{lemma}
Suppose the quantum plant satisfies the PR conditions (\ref{CCR11}), (\ref{CCR12_plant}) and the matrix $E_t$ is of full column rank.  Then the closed-loop system state CCR matrix $\Theta_0$ in (\ref{Theta0}) is preserved  if and only if the controller satisfies the PR conditions (\ref{CCR22}), (\ref{CCR12_cont}).
\end{lemma}
\begin{proof}
The ``if'' part of the lemma was considered above. The ``only if'' claim is established by left multiplying both sides of (\ref{CCR12}) by $(E_t^{\rT}E_t)^{-1}E_t^{\rT}$. This is valid  if $E_t$ is of full column rank and yields
$$
    c_t K_2 + d_t J_2 b_t^{\rT}
    =
    -
    (E_t^{\rT}E_t)^{-1}E_t^{\rT}
    (K_1 C_t^{\rT} + B_t J_1 D_t^{\rT}).
$$
In this case, the fulfillment of the second PR condition for the plant (\ref{CCR12_plant}) indeed entails the PR condition (\ref{CCR12_cont}) for the controller.
\end{proof}

The fact that the CCR preservation property for the closed-loop system state alone ``covers'' the separate PR conditions for the plant and controller as input-output operators is  explained by the ``internalization'' of their outputs which become part of the state dynamics when the systems are coupled.
In what follows, the controller state CCR matrix is assumed to be given by
\begin{equation}
\label{J0}
    K_2= I_{\nu}\ox \bJ =: J_0,
\end{equation}
where $\nu := n/2$,
so that the controller PR conditions (\ref{CCR22}) and (\ref{CCR12_cont}) take the form
\begin{align}
\label{PR1}
    a_tJ_0 +J_0a_t^{\rT}
    +
    e_t D_t J_1 D_t^{\rT}e_t^{\rT} + b_t J_2 b_t^{\rT}
    & =  0,\\
\label{PR2}
    c_t J_0 + d_t J_2 b_t^{\rT}& =  0.
\end{align}
%
%
The first PR condition  (\ref{PR1}) is a linear equation with respect to $a_t$ whose solutions are parameterized by real symmetric matrices $R_t$ of order $n$ as
\begin{equation}
\label{a}
    a_t
    =
    \underbrace{(e_t D_t J_1 D_t^{\rT}e_t^{\rT} + b_t J_2 b_t^{\rT})J_0/2}_
    {\wt{a}_t}
    \ +
    J_0 R_t.
\end{equation}
These solutions
form an affine subspace in $\mR^{n\x n}$ obtained by translating the linear subspace of Hamiltonian matrices
$$
    \{
        a \in \mR^{n\x n}:\
        a J_0 + J_0 a^{\rT} = 0
    \}
    =
    J_0 \mS_n
    =
    \mS_n J_0
$$
by a skew-Hamiltonian  matrix $\wt{a}_t$,
that is, a particular solution of (\ref{PR1}) which is a quadratic function
of $b_t$ and $e_t$.
Here, $\mS_n$ denotes the subspace of real symmetric matrices of order $n$, and $R_t\in \mS_n$ specifies the
free Hamiltonian operator $\xi_t^{\rT} R_t \xi_t/2$ of the quantum harmonic oscillator
\cite[Eqs. (20)--(22) on pp. 8--9]{EB_2005}. Note that (\ref{a}) is
the orthogonal decomposition of $a_t$ into projections onto the subspaces of skew-Hamiltonian and Hamiltonian matrices in the sense of  the Frobenius inner product
$$
    \bra X, Y\ket
    :=
    \Tr(X^{\rT}Y).
$$
Since the canonical structure of $J_0$ in (\ref{J0}) implies that $J_0^{-1} = -J_0$, the second PR condition (\ref{PR2}) allows the matrix $c_t$ to be expressed in terms of $b_t$ as
\begin{equation}
\label{c}
    c_t
    =
    d_t
    J_2
    b_t^{\rT}
    J_0.
\end{equation}
The matrix $d_t$, which quantifies the instantaneous gain of the controller output $\eta_t$ with respect to the controller noise $\omega_t$, is assumed to be fixed. Then (\ref{a}),  (\ref{c}) completely parameterize the state-space matrices of a PR controller by the matrix triple
\begin{equation}
\label{u}
    u_t
    :=
    (b_t,e_t, R_t),
\end{equation}
which will be regarded as an element of the Hilbert space $\mU:=\mR^{n\x m_2}\x \mR^{n\x p_1}\x \mS_n$ with the inherited inner product $\bra (b,e,R),(\beta,\epsilon,\rho) \ket := \bra b,\beta\ket+\bra e,\epsilon\ket +\bra R,\rho\ket$.

\section{Coherent quantum LQG control problem}\label{sec:problem}

In extending the infinite-horizon time invariant case from  \cite{NJP_2009,VP_2011}, the CQLQG control problem  is formulated as the minimization of the average output ``energy'' of the closed-loop system (\ref{closed}) over a bounded time interval $[0,T]$:
\begin{equation}
\label{E}
    \cE_T
    :=
        \int_{0}^{T}\!
        \bE
        (
            \cZ_t^{\rT} \cZ_t
        )
        \rd t
        =
    \int_{0}^{T}
    \bra
        \cC_t^{\rT}\cC_t,
        P_t
    \ket
    \rd t
    \longrightarrow
    \min,
\end{equation}
where the minimum is taken over the maps $t\mapsto u_t$ in (\ref{u}) which parameterize the   $n$-dimensional controllers (\ref{xi})--(\ref{zeta}) satisfying the PR conditions (\ref{PR1}), (\ref{PR2}). Here, $Z_t^{\rT} Z_t$ is the sum of squared entries of the vector $Z_t$  from (\ref{cZ}), which are self-adjoint quantum mechanical  operators. Also,
\begin{equation}
\label{P}
    P_t
    :=
    \Re
    \bE
    (
        \cX_t\cX_t^{\rT}
    )
\end{equation}
is a real positive semi-definite symmetric matrix of order $2n$ (we denote the set of such matrices by $\mS_{2n}^+$) which is the entrywise real part of the quantum covariance matrix of the closed-loop system state vector $\cX_t$. The matrix $P_t$ satisfies the differential Lyapunov equation
\begin{equation}
\label{Pdot}
    \dot{P}_t
    =
      \cA_t     P_t
    +
     P_t   \cA_t    ^{\rT}
    +
      \cB_t       \cB_t    ^{\rT}
      =:
    \cL_{t,u_t}(P_t).
\end{equation}
The affine operator $\cL_{t,u_t}$ is the infinitesimal generator of a two-parameter semi-group, which acts on $\mS_{2n}^+$ and depends on the triple $u_t$ of current matrices  of the PR controller from (\ref{u}) through (\ref{cABC}), (\ref{a}), (\ref{c}).
If $P_0$ were zero (which is forbidden by the positive semi-definiteness  of the quantum covariance matrix $\bE(\cX_0\cX_0^{\rT})=P_0 + i\Theta_0 /2\succcurlyeq 0$), then $P_t$ would coincide with the controllability
Gramian, over the time interval $[0,t]$,   of a classical linear time-varying system with the state-space realization triple $( \cA_t    ,  \cB_t     ,  \cC_t      )$ driven by a standard Wiener process.
The fact that $ \cE_T $ in (\ref{E}) is representable as the LQG cost of a classical system reduces the CQLQG problem to a constrained LQG control problem for an equivalent classical plant
\begin{equation}
\label{class_plant}
    \left[
        {\begin{array}{l||lc|c}
                  A_t     &   B_t   & E_t d_t & E_t\\
                \hline
                \hline
                  F_t    &   0   & 0 & G_t\\
                  \hline
                  0     &  0  & I & 0\\
                  C_t     &  D_t  & 0 & 0
        \end{array}}
    \right]
\end{equation}
driven by an $(m_1+m_2)$-dimensional standard Wiener process, with the controller being noiseless. In accordance with the standard convention, the block structure of the  state-space realization in (\ref{class_plant}) corresponds to partitioning the input into the noise and control, and the output into the to-be-controlled and observation signals.
We will develop a dynamic programming approach to (\ref{E}) as an optimal control problem for a subsidiary dynamical system with state $P_t$ in (\ref{P}) governed by the ODE  (\ref{Pdot}) whose right-hand side is specified by the matrix triple  $u_t$ from (\ref{u}). With the time horizon $T$ assumed to be fixed, the minimum cost  function is defined by
\begin{equation}
\label{V}
    V_t(P)
    :=
    \inf
    \int_{t}^{T}
        \bra
        \cC_s^{\rT}\cC_s,
        P_s
    \ket
    \rd s
\end{equation}
for any $t \in [0,T]$ and $P\in \mS_{2n}^+$.
Here, the infimum is taken over all admissible state-space matrices of the PR controller on the time interval $[t,T]$, provided the initial symmetric covariance matrix of the closed-loop system state vector is $\Re \bE(\cX_t\cX_t^{\rT})= P$.

\section{Symplectic invariance}\label{sec:symplectic}
As in the time-invariant case  \cite{VP_2011}, the PR conditions (\ref{PR1})--(\ref{PR2}) are invariant with respect to  the group of \textit{symplectic} similarity transformations of the controller matrices
$$
    a_t \mapsto \sigma a_t \sigma^{-1},
    \quad
    b_t \mapsto \sigma b_t,
    \quad
    e_t \mapsto \sigma e_t,
    \quad
    c_t \mapsto c_t\sigma^{-1},
$$
where $\sigma$ is an arbitrary (possibly, time-varying) symplectic matrix of order $n$ in the sense that  $\sigma J_0\sigma^{\rT} = J_0$. This corresponds to the canonical state transformation $\xi_t \mapsto \sigma \xi_t$; see also \cite[Eqs. (12)--(14)]{S_2000}.  Any such transformation of a PR controller leads to its  equivalent state-space representation, with the matrix $R_t$ transformed as $R_t\mapsto \sigma^{-\rT}R_t \sigma^{-1}$, where
$(\cdot)^{-\rT}:= ((\cdot)^{-1})^{\rT}$. 
Hence,
the minimum cost function $V_t(P)$ in (\ref{V}) is invariant under the corresponding group of  transformations of the closed-loop system state covariance matrix $P$, that is,
 \begin{equation}
\label{VV}
    V_t(S P S^{\rT})
    =
    V_t(P),
    \qquad
    S
    :=
    \begin{bmatrix}
        I & 0\\
        0 & \sigma
    \end{bmatrix}
\end{equation}
for any symplectic matrix $\sigma$.
 Assuming that $V_t(P)$ is Frechet differentiable in  $P$, its \textit{symplectic invariance} (\ref{VV}) can be described in differential terms. To formulate the lemma below, the matrix
$P\in \mS_{2n}^+$ is split into blocks as
\begin{equation}
\label{Pblocks}
    P
    :=
    {\begin{array}{cc}
    {}_{\leftarrow n \rightarrow}  {}_{\leftarrow n\rightarrow} &\\
            {\begin{bmatrix}
                P_{11} & P_{12}\\
                P_{21} & P_{22}
            \end{bmatrix}}
    &\!\!
        {\begin{matrix}
            \updownarrow\!{}^n\\
            \updownarrow\!{}_n
        \end{matrix}}\\
        {}
    \end{array}}
    =
    {\begin{array}{cc}
    {}_{\leftarrow n \rightarrow}  {}_{\leftarrow n\rightarrow} &\\
            {\begin{bmatrix}
                P_{\bullet 1} & P_{\bullet 2}
            \end{bmatrix}}
    &\!\!
            \updownarrow^{2n}
        \\ {}
    \end{array}}
    =
    {\begin{array}{cc}
    {}_{\leftarrow 2n \rightarrow}\\
            {\begin{bmatrix}
                P_{1\bullet}\\
                P_{2\bullet}
            \end{bmatrix}}
    &\!\!
        {\begin{matrix}
            \updownarrow\!{}^n\\
            \updownarrow\!{}_n
        \end{matrix}}\\
        {}
    \end{array}},
\end{equation}
where $P_{11}$ is associated with the state variables of the
plant, whilst $P_{22}$ pertains to those of the controller.
The $\mS_{2n}$-valued Frechet derivative of the minimum cost function has an analogous partitioning
\begin{equation}
    Q_t(P)
     :=
     \d_PV_t(P)
     =
\label{Q}
    {\begin{bmatrix}
        \d_{P_{11}} V_t      & \d_{P_{12}} V_t/2\\
        \d_{P_{21}} V_t/2    &  \d_{P_{22}} V_t
    \end{bmatrix}},
\end{equation}
where the $1/2$-factor takes into account the symmetry of $P$. Note that $Q_T\equiv 0$ since $V_T\equiv 0$ in view of (\ref{V}).  Associated with $V_t$ is a map $H_t: \mS_{2n}^+ \to \mR^{2n \x 2n}$ defined by
\begin{equation}
\label{H}
    H_t(P)
    :=
    Q_t(P)P,
\end{equation}
which is also partitioned into blocks as in (\ref{Pblocks}) except that   the matrix $H_t(P)$ is not necessarily symmetric. Also,
\begin{equation}
\label{bH}
    \bH(N)
     :=
     -J_0\bS(J_0 N)
     =
     -\bS(NJ_0)J_0
     =
    (N+J_0N^{\rT}J_0)/2
\end{equation}
denotes the orthogonal projection of a matrix $N \in \mR^{n\x n}$ onto the subspace of Hamiltonian matrices, with $\bS$ the \textit{symmetrizer}  defined by
\begin{equation}
\label{bS}
    \bS(N)
    :=
    (N+N^{\rT})/2.
\end{equation}


\begin{lemma}
\label{lem:VPDE}
Suppose the minimum cost function $V_t(P)$ in (\ref{V}) is Frechet differentiable with respect to  $P\in \mS_{2n}^+$. Then it satisfies the PDE
\begin{equation}
\label{VPDE}
    \bH(
        H_t^{22}(P)
    )
    =
    0,
\end{equation}
which means that the controller block of the matrix $H_t(P)$ from (\ref{H}) is a skew-Hamiltonian matrix.
\end{lemma}
\begin{proof}
The transformations $P\mapsto S P S^{\rT}$ in (\ref{VV}) form a Lie group 
whose tangent space can be identified with the subspace of Hamiltonian matrices $\tau$ of order $n$.  As the matrix exponential of a Hamiltonian matrix, $\sigma_{\veps}:= \re^{\veps \tau}$ is a symplectic matrix for any real $\veps$. Therefore, if $V_t$ is smooth, then by differentiating the left-hand side of (\ref{VV}) with $S_{\veps}:= {\scriptsize\begin{bmatrix} I & 0\\ 0 & \sigma_{\veps}\end{bmatrix}}$ as a composite function at $\veps=0$, it follows that the resulting Lie derivative  \cite{G_1977} of $V_t$ vanishes:
\begin{align}
\nonumber
    0
    &=
    \left.
    \d_{\veps}
    V_t(S_{\veps} P S_{\veps}^{\rT})
    \right|_{\veps = 0}
    =
    2
    \Bra
        \d_PV_t,
        {\scriptsize\begin{bmatrix}
        0\\
        I
        \end{bmatrix}}
        \tau
        P_{2\bullet}
    \Ket\\
\label{Lieder}
    &=2
    \Bra
        {\begin{bmatrix}
        0 &         I
        \end{bmatrix}}
        Q_t
        P_{\bullet 2},
        \tau
    \Ket
    \,=
    2
    \bra
        \bH(H_t^{22}(P)),
        \tau
    \ket.
\end{align}
Here, the notations  (\ref{Pblocks})--(\ref{H}) are used.
Since the relation (\ref{Lieder}) is valid for any Hamiltonian matrix $\tau$, then  (\ref{VPDE}) follows.
\end{proof}
%
%

The relation (\ref{VPDE}) is, in fact, a system of first order scalar homogeneous linear  PDEs which are associated with $n(n+1)/2$ entries of a real symmetric matrix of order $n$. This system is underdetermined since,  for a fixed $P_{11}$,  the total number of independent scalar variables is $n(3n+1)/2$. As we will show in Appendix~\ref{sec:Vshape_proof},  this system of PDEs satisfies the involutivity condition and is locally completely integrable by the Frobenius integration theorem  \cite{G_1977}.
Instead of ``disassembling'' (\ref{VPDE}) into scalar equations which would lead to the loss of the underlying algebraic structure, it can be treated as one PDE  with noncommutative matrix-valued variables.
Such PDEs were encountered, for example,  in entropy variational problems for Gaussian diffusion processes \cite{VP_2010a}. 
The general solution of the PDE (\ref{VPDE}) is given below.
%
%

\begin{theorem}
\label{th:Vshape}
Suppose $f_t: \mS_n^+\x \mR^{n\x n}\to \mR$ is a continuously Frechet differentiable function. Then the function
\begin{equation}
\label{Vshape}
    V_t(P)
    :=
    f_t(P_{11}, P_{12}(P_{22}^{-1} + J_0)P_{21}),
\end{equation}
defined for $P\in \mS_{2n}^+$ with $P_{22} \succ 0$, satisfies the PDE (\ref{VPDE}). Moreover, (\ref{Vshape}) describes  a general smooth solution of the PDE over any connected component of the set $\{P\in \mS_{2n}^+:\, \det P_{12} \ne 0,\, P_{22} \succ 0\}$.
\end{theorem}

The proof of Theorem~\ref{th:Vshape} is given in Appendix \ref{sec:Vshape_proof} and employs ideas from the method of characteristics for conventional PDEs  \cite{E_1998,V_1971} in combination with a nonlinear change of variables. The theorem shows that, due to the symplectic invariance,   which holds at any time $t$, the minimum cost function $V_t(P)$ depends  on the matrix $P$ only through the special combinations of its blocks $P_{11}$, $P_{12}P_{22}^{-1}P_{21}$,  $P_{12}J_0 P_{21}$ which constitute a maximal set of nonconstant invariants of $P$ with respect to the transformation group $P\mapsto SPS^{\rT}$ described in (\ref{VV}).

\section{Hamilton-Jacobi-Bellman equation}\label{sec:HJBE}

Assuming that the minimum cost  function $V_t(P)$ from (\ref{V}) is continuously differentiable with respect to $t$ and $P$ in the sense of Frechet,  the dynamic programming principle yields the HJBE
\begin{equation}
\label{HJBE}
    \d_t V_t(P)
    +
    \inf_{u \in \mU}
    \Pi_t(P,u,Q_t(P))
    = 0.
\end{equation}
Here, the minimization is over the triple $u:=(b,e,R)$ of the current matrices $b:= b_t$, $e:= e_t$, $R:= R_t$ of the PR controller, and the map $Q_t: \mS_{2n}^+ \to \mS_{2n}$ is associated with $V_t$ by (\ref{Q}). Also,
\begin{equation}
\label{Pi}
    \Pi_t(P,u,Q)
    :=
            \bra
                \cC_t^{\rT}\cC_t,
                P
            \ket
            +
            \bra
                Q,
                \cL_{t,u}(P)
            \ket
\end{equation}
is the \textit{control Hamiltonian}  \cite{SW_1997} of Pontryagin's minimum principle  \cite{PBGM_1962}  applied to the CQLQG problem (\ref{E}) as an optimal control  problem for the dynamical system (\ref{Pdot}) with state $P$ and control $u$. In view of (\ref{cABC}), (\ref{a}), (\ref{c}) and (\ref{Pdot}),  the matrix $R \in \mS_n$, which parameterizes the free Hamiltonian operator of the PR controller, enters the control Hamiltonian $\Pi_t(P,b,e,R,Q)$ only through $\cA_t$ and in an affine fashion. Moreover, by considering (\ref{Pi}) with $Q= Q_t(P)$ as in (\ref{HJBE}), it follows that
\begin{align}
\nonumber
    \Pi_t(P,b,e,R,Q_t(P))
    &=
    \Pi_t(P,b,e,0,Q_t(P))
    +
        2\Bra
            Q_t(P),
            \begin{bmatrix}
                0\\
                I
            \end{bmatrix}
            J_0 R
            \begin{bmatrix}
                0 & I
            \end{bmatrix}
                P
        \Ket        \\
\nonumber
         &=
\Pi_t(P,b,e,0,Q_t(P))   +
    2
        \Bra
            \bH(H_t^{22}(P)),
            J_0R
        \Ket\\
\label{Rdep0}
    &=
    \Pi_t(P,b,e,0,Q_t(P)),
\end{align}
where we have used the notations (\ref{Pblocks})--(\ref{H}) and Lemma~\ref{lem:VPDE}. The $R$-independence of the right-hand side of  (\ref{Rdep0}) reduces the minimization problem in (\ref{HJBE}) to
\begin{equation}
\label{red0}
    \inf_{u \in \mU}
    \Pi_t(P,u,Q_t(P))
    =
    \inf_{b,e}
    \Pi_t(P,b,e,0,Q_t(P)).
\end{equation}
This does not mean, however, that  $R_t=0$ has to be satisfied for the optimal quantum controller.  The optimization problem (\ref{red0}) is solved in Section~\ref{sec:gain}.
We will now show that the map $Q_t$ from (\ref{Q}), evaluated at an optimal trajectory of the system (\ref{Pdot}) and thus describing the costate of this system 
through the Pontryagin equations
\begin{equation}
\label{PQdot}
    \dot{P}_t
    =
    \d_Q \Pi_t,
    \qquad
    \dot{Q}_t
    =
    -\d_P \Pi_t,
\end{equation}
coincides with the observability Gramian of the closed-loop system (\ref{cABC}). Here, $\dot{(\, )}$ is the total time derivative, and the partial Frechet derivatives of (\ref{Pi}) are taken with respect to $Q$, $P$  as independent $\mS_{2n}$-valued variables.

%

\begin{lemma}
\label{lem:QODE}
Suppose the minimum cost  function $V_t(P)$ in (\ref{V}) is twice continuously Frechet differentiable in $t$ and $P$. Also, suppose there exist functions $b_t^{\diam}(P)$ and $e_t^{\diam}(P)$ which are Frechet differentiable  in $P$ and deliver the minimum in (\ref{red0}).
 Then the matrix  $Q_t^{\diam}:= Q_t(P_t^{\diam})$, obtained  by evaluating the map (\ref{Q}) at an optimal trajectory $P_t^{\diam}$ of the system (\ref{Pdot}), satisfies the differential Lyapunov equation
\begin{equation}
\label{Qdot}
    \dot{Q}_t^{\diam}
    =
    -{\cA_t^{\diam}}^{\rT} Q_t^{\diam}
    -
    Q_t^{\diam} \cA_t^{\diam}
    -
    {\cC_t^{\diam}}^{\rT}\cC_t^{\diam},
\end{equation}
where $\cA_t^{\diam}$, $\cC_t^{\diam}$ are the corresponding state-space matrices of the closed-loop system with the optimal CQLQG controller.
\end{lemma}
\begin{proof}
The twice continuous differentiability of $V_t(P)$  ensures the interchangeability of its partial derivatives in $t$ and $P$, so that
$\d_P \d_t V_t = \d_t \d_P V_t =\d_t Q_t$ in view of (\ref{Q}). Hence, by substituting (\ref{red0}) into (\ref{HJBE}) and differentiating the HJBE with respect to $P$, it follows that
\begin{equation}
\label{QPDE1}
    \d_t Q_t
    +
        \rd_P
        \Pi_t(P,b_t^{\diam},e_t^{\diam},0,Q_t)/\rd P
        =0.
\end{equation}
Since the pair $(b_t^{\diam},e_t^{\diam})$ minimizes $\Pi_t(P,b,e,Q_t)$ in $(b,e)\in \mR^{n\x m_2}\x \mR^{n\x p_1}$ (that is, over an open set), then it is a critical point of this function, where both $\d_b\Pi_t$ and $\d_e\Pi_t$ vanish. Therefore,
\begin{align}
\nonumber
    \rd\Pi_t(P,b_t^{\diam},e_t^{\diam},0,Q_t)/\rd P
    \!=&
    \d_P\Pi_t + (\d_P b_t^{\diam})^{\dagger}(\d_b\Pi_t)
    + (\d_P e_t^{\diam})^{\dagger}(\d_e\Pi_t) + (\d_P Q_t)^{\dagger}(\d_Q\Pi_t)\\
\label{dPi}
    \!=&
    \d_P \Pi_t+ (\d_P Q_t)^{\dagger}(\d_Q\Pi_t),
\end{align}
where $\d_PQ_t = \d_P^2 V_t$ is a well-defined self-adjoint operator on $\mS_{2n}$ in view of (\ref{Q}) and the twice continuous Frechet differentiability of $V_t$.
Now, from  (\ref{Pdot}) and (\ref{Pi}), it follows  that
\begin{equation}
\label{dPi1}
    \d_P\Pi_t
     =
    \cC_t^{\rT}\cC_t
    +
    (\d_P\cL_{t,u}(P))^{\dagger}(Q)
     =
    \cA_t^{\rT} Q
    +
    Q \cA_t
    +
     \cC_t^{\rT}\cC_t.
\end{equation}
Since (\ref{Pi}) implies that $\d_Q\Pi_t = \cL_{t,u}(P)$, then substitution of (\ref{dPi}),  (\ref{dPi1}) into (\ref{QPDE1}) yields the PDE
\begin{equation}
\label{QPDE}
    \d_t Q_t
    +
    {\cA_t^{\diam}}^{\rT} Q_t
    +
    Q_t \cA_t^{\diam}
    +
    {\cC_t^{\diam}}^{\rT}\cC_t^{\diam}
    +
    (\d_P Q_t)^{\dagger}(\cL_{t,u_t}(P))
    =0.
\end{equation}
For the matrix $P_t^{\diam}$ governed by (\ref{Pdot}) with $u_t =(b_t^{\diam},e_t^{\diam},R_t)$, the matrix
$
    (\d_t Q_t )(P_t^{\diam})
    +
    (\d_P Q_t)^{\dagger}(\cL_{t,u_t}(P_t^{\diam}))
    =
    \rd Q_t(P_t^{\diam})/\rd t
$
is the total time derivative of $Q_t^{\diam}$. Hence, (\ref{QPDE}) leads to (\ref{Qdot}).
\end{proof}
%

The ODE (\ref{Qdot}), whose right-hand side coincides with $-\d_P \Pi_t$ in view of (\ref{dPi1}) and in conformance with (\ref{PQdot}), is the differential Lyapunov equation which governs the observability Gramian $Q_t^{\diam}$ of the closed-loop system under the optimal CQLQG controller. It is solved backwards in time $t\<T$ with zero terminal condition $Q_T^{\diam}= 0$. This (or an alternative reasoning involving the monotonicity of $V_t(P)$ in  $P$), can be used to show that the map $Q_t$ given  by (\ref{Q}), takes values in $\mS_{2n}^+$.
Therefore, $H_t(P)$ in (\ref{H})
is a diagonalizable matrix with all real nonnegative eigenvalues which  correspond to the squared Hankel singular values of the closed-loop system in view of the interpretation of $Q_t$ and $P_t$  as observability and controllability Gramians. We will refer to $H_t$ as the \textit{Hankelian} of the closed-loop system.

\section{Optimal controller gain matrices}\label{sec:gain}

Since the matrix $\cC_t$ in (\ref{cABC})  depends only  on $b_t$ in view of (\ref{c}), the minimization on the right-hand side of (\ref{red0}) can be represented as
\begin{equation}
\label{minbe}
    \inf_u
    \Pi_t(P,u,Q_t)
    =
    \inf_{b}
    (
        \bra
            \cC_t^{\rT}\cC_t,
            P
        \ket
        +
        \inf_{e}
        \Bra
            Q_t,
            \cL_{t,b,e,0}(P)
        \Ket
    )
\end{equation}
which is a repeated minimization problem over the PR controller gain matrices $b:=b_t$ and $e:=e_t$.
Here,
\begin{equation}
\label{Ltbe0}
    \cL_{t,b,e,0}(P)
    =
    \wt{\cA}_t P + P\wt{\cA}_t^{\rT} + \cB_t \cB_t^{\rT}
\end{equation}
is the Lyapunov operator from (\ref{Pdot}) obtained by letting $R_t=0$ in (\ref{a}) and substituting the remaining skew-Hamiltonian part $\wt{a}_t$ of the controller matrix $a_t$ into (\ref{cABC}), which yields
\begin{equation}
\label{tcAt}
    \wt{\cA}_t
    :=
    \begin{bmatrix}
        A_t       & E_t c_t         \\
        e_t C_t   & \wt{a}_t
    \end{bmatrix}.
\end{equation}
In turn, the repeated minimization problem (\ref{minbe}) can be decoupled into two independent problems as follows. In view of the structure of the matrices $\wt{a}_t$ and $c_t$ in (\ref{a}) and (\ref{c}), the matrix $\wt{\cA}_t$ from (\ref{tcAt}) is a quadratic function of the controller gain matrices $b$, $e$. The dependencies of $\wt{\cA}_t$ on $b$ and $e$ can be isolated  as
\begin{equation}
\label{cAtilde}
    \wt{\cA}_t
    =
    \underbrace{
    {\begin{bmatrix}
     A_t & 0\\
        0   & 0
    \end{bmatrix}}}_{\cA_t^0}
    +
    \underbrace{\begin{bmatrix}
        0   &   E_t d_t J_2 b^{\rT} J_0\\
        0   &   b J_2 b^{\rT} J_0/2
    \end{bmatrix}}_{\breve{\cA}_t}
    +
    \underbrace{\begin{bmatrix}
        0        &  0\\
        eC_t     &    eD_t J_1 D_t^{\rT} e^{\rT} J_0/2
    \end{bmatrix}}_{\wh{\cA}_t},
\end{equation}
where the matrix $\cA_t^0$
is independent of both $b$ and $e$, whilst $\breve{\cA}_t$
only depends on $b$ and $\wh{\cA}_t$ only depends on $e$. In a similar vein, (\ref{cABC}) and (\ref{c}) imply that
\begin{align}
\label{BB}
    \cB_t\cB_t^{\rT}
    =&
        \underbrace{\begin{bmatrix}
            B_tB_t^{\rT} + E_td_td_t^{\rT}E_t^{\rT} &0\\
            0 & 0
        \end{bmatrix}}_{\Gamma_t^0}
    +
        \underbrace{\begin{bmatrix}
            0    &  E_t d_t b^{\rT}\\
            b d_t^{\rT}E_t^{\rT} &  bb^{\rT}
        \end{bmatrix}}_{\breve{\Gamma}_t}
    +
        \underbrace{\begin{bmatrix}
            0               & B_t D_t^{\rT} e^{\rT}  \\
            e D_t B_t^{\rT} & e D_t D_t^{\rT} e^{\rT}
        \end{bmatrix}}_{\wh{\Gamma}_t},\\
\label{CC}
    \cC_t^{\rT} \cC_t
    =&
    \underbrace{\begin{bmatrix}
        F_t^{\rT}F_t    & 0\\
        0               & 0
    \end{bmatrix}}_{\Delta_t^0}
    +
    \underbrace{\begin{bmatrix}
        0                               &  F_t^{\rT} G_t d_t J_2 b^{\rT}J_0\\
        J_0bJ_2d_t^{\rT}G_t^{\rT}F_t    &  J_0bJ_2d_t^{\rT}G_t^{\rT}G_t d_t J_2 b^{\rT}J_0
    \end{bmatrix}}_{\breve{\Delta}_t},
\end{align}
where $\Gamma_t^0$, $\Delta_t^0$ are independent of both $b$ and $e$, the matrices  $\breve{\Gamma}_t$, $\breve{\Delta}_t$
only depend on $b$, whilst $\wh{\Gamma}_t$ only depends on $e$.
%
By substituting (\ref{cAtilde}), (\ref{BB}) into (\ref{Ltbe0}) and combining the result  with (\ref{CC}),
the repeated minimization problem in (\ref{minbe}) is indeed split into
\begin{align}
\nonumber
    \inf_u
    \Pi_t(P,u,Q_t)
    =&
    \bra
        \Delta_t^0,P
    \ket
    +
    \bra
        Q_t,
        2 \cA_t^0 P + \Gamma_t^0
    \ket\\
\nonumber
    &+\inf_b
    (
        \bra
            \breve{\Delta}_t,
            P
        \ket
        +
        \bra
            Q_t,
            2 \breve{\cA}_t P + \breve{\Gamma}_t
        \ket
    )\\
\label{minb_mine}
    &+\inf_e
        \bra
            Q_t,
            2 \wh{\cA}_t P + \wh{\Gamma}_t
        \ket
    ).
\end{align}
Both minimization problems on the right-hand side of (\ref{minb_mine}) are quadratic optimization problems whose solutions are available in closed form and lead to the optimal values for the controller gain matrices $b$ and $e$.  The fact that these problems are independent describes a \textit{quasi-separation} property of the gain matrices  \cite{VP_2011} and can be interpreted as a weaker quantum counterpart of the filtering/control separation principle of the classical LQG control.  We will first consider the minimization with respect to  the controller observation gain matrix $e$. From (\ref{cAtilde}) and (\ref{BB}),  it follows that
\begin{align}
\nonumber
        \bra
            Q_t,
            2 \wh{\cA}_t P + \wh{\Gamma}_t
        \ket
        &=
        \Bra
            Q_t,
            {\begin{bmatrix}
                0       &   0\\
                2eC_t     &    eD_t J_1 D_t^{\rT} e^{\rT} J_0
            \end{bmatrix}}
            P
            +
        {\begin{bmatrix}
            0               &  B_t D_t^{\rT} e^{\rT}  \\
            e D_t B_t^{\rT} &  e D_t D_t^{\rT} e^{\rT}
        \end{bmatrix}}
        \Ket \\
\label{mine}
        &=
        \bra
            2(H_t^{21} C_t^{\rT} + Q_t^{21} B_t D_t^{\rT}) + \fM_t(e),
            e
        \ket,
\end{align}
where
\begin{equation}
\label{fM}
    \fM_t
    :=
    \[[[
            H_t^{22}
            J_0,
        D_t J_1 D_t^{\rT}
            \mid
        Q_t^{22},
        D_tD_t^{\rT}
    \]]]
\end{equation}
is a self-adjoint operator of grade two (see Appendix~\ref{sec:class}) on $\mR^{n\x p_1}$. Here, we have used the property that the matrix $H_t^{22}J_0$  is antisymmetric since the controller block $H_t^{22}$ of the Hankelian (\ref{H}) is skew-Hamiltonian in view of (\ref{VPDE}).   If the operator $\fM_t$ is positive definite, then the quadratic function  on the right-hand side of (\ref{mine}) achieves its minimum value
\begin{equation}
\label{emin}
    \min_e
        \bra
            Q_t,
            2 \wh{\cA}_t P + \wh{\Gamma}_t
        \ket
        =
        -
        \|
            H_t^{21} C_t^{\rT} + Q_t^{21} B_t D_t^{\rT}
        \|_{\fM_t^{-1}}^2
\end{equation}
at a unique point
\begin{equation}
\label{ediam}
    e_t^{\diam}
    :=
    -\fM_t^{-1}(H_t^{21} C_t^{\rT} + Q_t^{21} B_t D_t^{\rT}).
\end{equation}
Here, for a positive definite self-adjoint operator $\fO$ on the Hilbert space $\mR^{p\x q}$ with the standard Frobenius inner product $\bra\cdot, \cdot\ket$, we denote by
$$
    \|N\|_{\fO}
    :=
    \sqrt{\bra N, N\ket_{\fO}}
$$
the norm of a matrix $N\in \mR^{p\x q}$ associated with the ``weighted'' Frobenius inner product
$$
    \bra K, N\ket_{\fO}
    :=
    \bra K, \fO(N)\ket.
$$
The minimization in (\ref{minb_mine}) with respect to the controller noise gain matrix $b$ is performed in a similar fashion. It follows from (\ref{cAtilde})--(\ref{CC}) that
\begin{align}
\nonumber
        \bra&
            \breve{\Delta}_t,
            P
        \ket
        +
        \bra
            Q_t,
            2 \breve{\cA}_t P + \breve{\Gamma}_t
        \ket
\\
\nonumber
    =&
    2
    \bra
        P_{21},
        J_0 b J_2 d_t^{\rT} G_t^{\rT} F_t
    \ket
    +
    \bra
        P_{22},
        J_0 bJ_2 d_t^{\rT}G_t^{\rT}G_t d_t J_2 b^{\rT}J_0
    \ket\\
\nonumber
    &+
    \Bra
        Q_t,
        {\begin{bmatrix}
            0   &      2E_t d_t J_2 b^{\rT} J_0\\
            0   &      b J_2 b^{\rT} J_0
        \end{bmatrix}}
        P
        +
        {\begin{bmatrix}
            0    &  E_t d_t b^{\rT}\\
            b d_t^{\rT}E_t^{\rT} &  bb^{\rT}
        \end{bmatrix}}
    \Ket\\
\label{bmin0}
    =&
    \Bra
        2(
            Q_t^{21} E_t d_t
            +
            J_0
            (
                (H_t^{12})^{\rT} E_t
                +
                P_{21} F_t^{\rT} G_t
            )
             d_t J_2
        )
        +
        \fN_t(b),
        b
    \Ket,
\end{align}
where
\begin{equation}
\label{fN}
    \fN_t
     :=
    \[[[
            H_t^{22} J_0,
        J_2
            \mid
        Q_t^{22},
        I
            \mid
        J_0 P_{22} J_0,
        J_2 d_t^{\rT} G_t^{\rT} G_t d_t J_2
    \]]]
\end{equation}
is a self-adjoint operator of grade three (see Appendix~\ref{sec:class}) on $\mR^{n\x m_2}$ in view of the antisymmetry of $H_t^{22}J_0$.
 If $\fN_t$ is positive definite, then the quadratic function of $b$, given by (\ref{bmin0}), achieves its minimum value
\begin{align}
\nonumber
    &\min_b
    (
        \bra
            \breve{\Delta}_t,
            P
        \ket
        +
        \bra
            Q_t,
            2 \breve{\cA}_t P + \breve{\Gamma}_t
        \ket
    )\\
\label{bmin}
    &=
    -
    \|
            Q_t^{21} E_t d_t
            +
            J_0
            (
                (H_t^{12})^{\rT} E_t
                +
                P_{21} F_t^{\rT} G_t
            )
             d_t J_2
    \|_{\fN_t^{-1}}^2
\end{align}
at a unique point
\begin{equation}
\label{bdiam}
    b_t^{\diam}
    :=
    -\fN_t^{-1}(
        Q_t^{21} E_t d_t
            +
            J_0
            (
                (H_t^{12})^{\rT} E_t
                +
                P_{21} F_t^{\rT} G_t
            )
             d_t J_2
             ).
\end{equation}
Finally, by substituting (\ref{emin}), (\ref{bmin}) into (\ref{minb_mine}) and using the representation
$$
    \bra
        \Delta_t^0,P
    \ket
    +
    \bra
        Q_t,
        2 \cA_t^0 P + \Gamma_t^0
    \ket=    \bra
        F_t^{\rT}F_t,
        P_{11}
    \ket
    +
    2
    \bra
        H_t^{11},
        A_t
    \ket
    +
    \bra
        Q_t^{11},
        B_tB_t^{\rT}
        +
        E_td_td_t^{\rT}E_t^{\rT}
    \ket,
$$
which follows from (\ref{cAtilde})--(\ref{CC}),  the HJBE (\ref{HJBE}) is reduced to the Hamilton-Jacobi equation (HJE) below.

\begin{theorem}
\label{th:HJE}
Suppose the minimum cost  function $V_t(P)$ for the CQLQG problem, defined by (\ref{V}), is continuously Frechet differentiable in $t$ and $P$, and the associated self-adjoint   operators $\fM_t$ and $\fN_t$ in (\ref{fM}) and (\ref{fN}) are positive definite. Then the function $V_t(P)$ satisfies the HJE
\begin{align}
\nonumber
    \d_t V_t
    &+
   \bra
        F_t^{\rT}F_t,
        P_{11}
    \ket
    +
    2
    \bra
        H_t^{11},
        A_t
    \ket
    +
    \bra
        Q_t^{11},
        B_tB_t^{\rT}
        +
        E_td_td_t^{\rT}E_t^{\rT}
    \ket\\
\nonumber
    &-
    \|
            Q_t^{21} E_t d_t
            +
            J_0
            (
                (H_t^{12})^{\rT} E_t
                +
                P_{21} F_t^{\rT} G_t
            )
             d_t J_2
    \|_{\fN_t^{-1}}^2\\
\nonumber
        &-
        \|
            H_t^{21} C_t^{\rT} + Q_t^{21} B_t D_t^{\rT}
        \|_{\fM_t^{-1}}^2
        =0,
\end{align}
and the  gain matrices $b_t^{\diam}$, $e_t^{\diam}$  of an optimal PR controller are computed according to (\ref{bdiam}) and (\ref{ediam}).
\end{theorem}

By using \cite[Lemma~5]{VP_2011}, it can be shown, that if the controller block $Q_t^{22}$ of the closed-loop system observability Gramian is nonsingular, the matrix $D_t$ is of full row rank and $\br((Q_t^{22})^{-1}H_t^{22}J_0)<1$, with $\br(\cdot)$ the spectral radius of a matrix, then both operators $\fM_t$ and $\fN_t$
are positive definite. Since each of the matrices $Q_t(P)$ and  $P$ enters $\fM_t$ and $\fN_t$ in (\ref{fM}) and (\ref{fN}) in a linear fashion, the dependence of $b_t^{\diam}$ and $e_t^{\diam}$ on these  matrices is linear-fractional and hence, smooth, provided $\fM_t\succ 0$ and $\fN_t\succ 0$ (such values of $P$ form an open set).  Therefore, if, in addition to the assumptions of Theorem~\ref{th:HJE},  the minimum  cost function $V_t(P)$ is \textit{twice}  continuously Frechet differentiable with respect to $P$, then the optimal CQLQG controller gain matrices $b_t^{\diam}$ and $e_t^{\diam}$ are continuously Frechet differentiable functions of $P$. If, furthermore, $V_t$ is twice continuously Frechet differentiable in $t$ and $P$, this ensures the applicability of Lemma~\ref{lem:QODE}, which utilizes the viewpoint of Pontryagin's minimum principle on the CQLQG problem.

\section{Equations for the optimal quantum controller}\label{sec:eqsummary}

The set of equations for the optimal CQLQG controller  over the time interval $0\< t \< T$   consists of two  Lyapunov ODEs (\ref{Pdot}) and (\ref{Qdot}) for the controllability and observability Gramians $P_t$, $Q_t$ of the closed-loop system:
\begin{align}
\label{Pdot1}
        \dot{P}_t
    =&
      \cA_t     P_t
    +
     P_t   \cA_t    ^{\rT}
    +
      \cB_t       \cB_t    ^{\rT},\\
\label{Qdot1}
        \dot{Q}_t
    =&
      -\cA_t^{\rT}     Q_t
    -
     Q_t   \cA_t
    -
      \cC_t^{\rT}       \cC_t,
\end{align}
with the split boundary conditions $P_0= P$ and $Q_T=0$, where $P\in\mS_{2n}^+$ is a given matrix satisfying $P+ i\Theta_0/2\succcurlyeq 0$. According to (\ref{cABC}), (\ref{a}), (\ref{c}), the closed-loop system matrices $\cA_t$, $\cB_t$, $\cC_t$ are expressed in terms of the controller matrices $b_t$, $e_t$, $R_t$ as
\begin{align}
\label{cAdiam}
    \cA_t
    :=&
    {
        \begin{bmatrix}
                  A_t       & & E_t     d_t    J_2    b_t^{\rT}    J_0\\
                  e_t C_t   & & (e_t D_t J_1 D_t^{\rT}e_t^{\rT} + b_t J_2 b_t^{\rT})J_0/2 + J_0R_t\\
        \end{bmatrix}
    },\\
\label{cBdiam}
    \cB_t
    :=&
    {
        \begin{bmatrix}
            B_t       & & E_td_t\\
            e_t D_t   & & b_t
        \end{bmatrix}
    },\\
\label{cCdiam}
    \cC_t
    :=&
    {
        \begin{bmatrix}
F_t & & G_t     d_t
    J_2
    b_t^{\rT}
    J_0        \end{bmatrix}
    }.
\end{align}
In turn, the optimal controller gain matrices $b_t$, $e_t$ are completely specified by the Gramians $P_t$, $Q_t$ (which determine the closed-loop system Hankelian $H_t$) according to (\ref{ediam}), (\ref{bdiam}) as
\begin{align}
\label{ediam1}
    e_t
    :=&
    -
    \[[[
            H_t^{22}
            J_0,
        D_t J_1 D_t^{\rT}
            \mid
        Q_t^{22},
        D_tD_t^{\rT}
    \]]]^{-1}
    (H_t^{21} C_t^{\rT} + Q_t^{21} B_t D_t^{\rT}),\\
\nonumber
    b_t
    :=&
    -    \[[[
            H_t^{22} J_0,
        J_2
            \mid
        Q_t^{22},
        I
            \mid
        J_0 P_{22} J_0,
        J_2 d_t^{\rT} G_t^{\rT} G_t d_t J_2
    \]]]^{-1}\\
\label{bdiam1}
    &(
        Q_t^{21} E_t d_t
            +
            J_0
            (
                (H_t^{12})^{\rT} E_t
                +
                P_{21} F_t^{\rT} G_t
            )
             d_t J_2
             ),
\end{align}
where the inverses of the special self-adjoint operators can be  represented through the vectorization of matrices; see Appendix~\ref{sec:class}.  Therefore, the set of equations for the optimal CQLQG controller is a split boundary value problem for two Lyapunov ODEs (\ref{Pdot1}), (\ref{Qdot1}) which are nonlinearly coupled through the algebraic equations (\ref{cAdiam})--(\ref{bdiam1}). The matrix $R_t$, which affinely enters the right-hand side of these ODEs through the matrix $\cA_t$ in (\ref{cAdiam}), appears to be a free parameter in the sense that an equation for its optimal value is missing and the optimal controller gain matrices $e_t$ and $b_t$ in (\ref{ediam1}) and (\ref{bdiam1}) do not depend on the current value of $R_t$.
Moreover, by using the identity
\begin{align*}
\nonumber
    \dot{H}_t
    =&
    \dot{Q}_t P_t + Q_t \dot{P}_t\\
\nonumber
    =&
    -(\cA_t^{\rT} Q_t+Q_t\cA_t + \cC_t^{\rT}\cC_t) P_t
    +
    Q_t(\cA_t P_t+P_t\cA_t^{\rT} + \cB_t\cB_t^{\rT})\\
    =&
    [H_t, \cA_t^{\rT}] + Q_t \cB_t \cB_t^{\rT} - \cC_t^{\rT} \cC_t P_t
\end{align*}
(see  also \cite[Appendix C]{VP_2010b}), it can be shown that the skew-Hamiltonian structure of $H_t^{22}$ in (\ref{VPDE}), trivially ensured at $T=0$ by the terminal condition $Q_T=0$,  is preserved by  the dynamics (\ref{Pdot1})--(\ref{bdiam1})  for $t < T$ regardless of the choice of $R_t$. However, the function $[0,T]\ni t\mapsto R_t\in \mS_n$ is responsible for the fulfillment of the split boundary conditions.

\section{Concluding remarks}\label{sec:conclusion}

We have considered a time-varying Coherent Quantum LQG control problem which seeks a physically realizable quantum controller to minimize the finite-horizon LQG cost, and outlined a novel approach towards its solution. Using the Hamiltonian parameterization of PR controllers, which relates them to open quantum harmonic oscillators,  we have recast    the CQLQG problem
as a covariance control problem.  Dynamic programming and Pontryagin's minimum principle have been applied to the resulting optimal control problem for a subsidiary  deterministic dynamical system whose state is the symmetric  part of the quantum covariance matrix of the closed-loop system state vector governed by a differential Lyapunov equation. It has been shown that the corresponding costate is the observability Gramian of the closed-loop system. By using the invariance of the minimum cost function under the group of symplectic similarity transformations of PR controllers, we have derived algebraic equations for the gain matrices of the optimal CQLQG controller and established their partial decoupling as a weaker quantum analogue of the classical LQG control/filtering  separation principle. These equations express the optimal controller gain matrices in terms of the current observability and controllability Gramians of the closed-loop system thus leading to a split boundary value problem for two nonlinearly coupled differential Lyapunov equations. The difficulty of solving this problem lies in the coupling of the differential equations and mixed nature of the boundary conditions. However, the  special structure of the minimum cost function, enforced  by the symplectic invariance,   suggests the possibility of reducing the order of these equations by nonlinear transformation of the blocks of the Gramians. Another resource yet to be explored  is to consider the CQLQG problem for PR plants.
The existence/uniqueness of solutions to the equations for the state-space realization matrices of the optimal CQLQG controller remains an open problem and so do their possible reduction and numerical implementation. These issues are a subject of current research and will be reported in subsequent publications.






\appendix

\section{Proof of Theorem \ref{th:Vshape}}\label{sec:Vshape_proof}

By omitting the dependence of the minimum cost function $V_t(P)$ on $t$ and $P_{11}$ which are assumed to be fixed, and introducing the variables
\begin{equation}
\label{XY}
    X
    :=
    P_{12} = P_{21}^{\rT},
    \qquad
    Y
    :=
    P_{22},
\end{equation}
the PDE (\ref{VPDE}) takes the form
\begin{equation}
\label{MVJ}
    \bH(
        M(V)
    )
    = 0.
 \end{equation}
Here, use is made of (\ref{Pblocks})--(\ref{bS}), and $M$ is a linear differential operator which maps a Frechet differentiable function $\mR^{n\x n}\x \mS_n \ni (X,Y) \mapsto v(X,Y)\in \mR$ to an $\mR^{n\x n}$-valued function $M(v)$ defined on the same domain  by
\begin{equation}
\label{M}
    M(v)
    :=
    \frac{1}{2}
    X^{\rT}\d_X v
    +
    Y\d_Y v.
\end{equation}
The next section verifies the involutivity of the PDE (\ref{MVJ}) as a system of scalar PDEs. Then we consider two particular solutions of this PDE in Section~\ref{sec:two} which allow its  general solution to be obtained in Section~\ref{sec:gen} through a change-of-variables technique.

\subsection{Verification of involutivity}\label{sec:Frobenius}

We will now verify the fulfillment of the local complete integrability conditions for the PDE (\ref{MVJ}),  which in view of (\ref{bH}) and (\ref{bS}) is equivalent  to
\begin{equation}
\label{SMVJ}
    \bS(M(V)J_0) = 0
\end{equation}
whose left-hand side is a real symmetric matrix of order $n$. If the complete integrability holds, then the PDE has a $n^2$-dimensional integral manifold which can be represented using an $\mR^{n\x n}$-valued function of the matrices $X$ and $Y$. For any constant matrix $Z \in \mS_n$, let $L_Z$ be a linear operator which maps a Frechet differentiable function $\mR^{n\x n}\x \mS_n\ni (X,Y)\mapsto v(X,Y)\in \mR$ to a function $L_Z(v): \mR^{n\x n}\x \mS_n\to \mR$ defined by
\begin{equation}
\label{LZ}
    L_Z(v)
     =
    \bra
        Z,
        \bS(M(v)J_0)
    \ket
     =
    -
    \bra
        Z J_0,
        M(v)
    \ket
    =
    -
    \Lambda_{ZJ_0}(v).
\end{equation}
Here, $M$ is given by (\ref{M}),
and $\Lambda_N$ is a linear differential operator, which is associated  with a matrix $ N \in \mR^{n\x n}$ and maps the function $v$ to another function of $X$, $Y$ as
\begin{equation}
\label{Lambda}
    \Lambda_N(v)
    :=
    \bra
        N,
        M(v)
    \ket.
\end{equation}
Thus, the operators $L_Z$, associated with  symmetric matrices $Z$ by (\ref{LZ}),  are the operators $\Lambda_N$ considered for Hamiltonian matrices $N$, although, in general,  the matrix $N$ in (\ref{Lambda}) can be arbitrary.
By the Frobenius integration theorem \cite{G_1977} (see also, \cite[pp. 158--165]{B_1986}),
the local complete integrability of the PDE (\ref{SMVJ}) will be proved if we show that for any constant matrices $Z_1, Z_2 \in \mS_n$, there exists a matrix $Z_3\in \mS_n$, which is allowed to depend on $X$ and $Y$ and such that
\begin{equation}
\label{LLL}
    [
        L_{Z_1}, L_{Z_2}
    ](v)
    =
    L_{Z_3}(v)
\end{equation}
is satisfied for any twice continuously  Frechet  differentiable function $v$ described above. The relation (\ref{LLL}) is an inner product form of the involutivity condition for the PDE (\ref{SMVJ}) regarded as a system of scalar PDEs, with the inner product used to represent linear combinations of the individual equations in a coordinate-free fashion.
\begin{lemma}
\label{lem:Frob}
For any constant matrices $N_1, N_2\in \mR^{n\x n}$, the operators (\ref{Lambda}), considered on  twice continuously Frechet  differentiable test functions $v$,  satisfy the commutation relation
\begin{equation}
\label{LambdaLambda}
    [
        \Lambda_{N_1},
        \Lambda_{N_2}
    ]
    =
    \Lambda_{[N_1,N_2]/2}.
\end{equation}
\end{lemma}
\begin{proof}
For any $N\in \mR^{n \x n}$ and any twice continuously Frechet  differentiable function $v:\mR^{n \x n}\x \mS_n \mapsto \mR$, the function $\Lambda_N(v)$ is continuously Frechet differentiable and its derivatives are
computed as
\begin{align*}
    \d_X
    \Lambda_N(v)
     =&
    \frac{1}{2}
    (
        \d_X v N^{\rT}
        +
        \d_X^2 v(XN)
    ) +
        \d_Y \d_X v(YN),\quad\\
\nonumber
    \d_Y
    \Lambda_N(v)
     =&
    \frac{1}{2}
    (
        N \d_Y v + \d_Y vN^{\rT}
        +
        \d_Y^2 v(YN+N^{\rT}Y)
        +
        \d_X\d_Y v(XN)
    ),
\end{align*}
where the relation $\d_X\d_Y v = (\d_Y\d_X v)^{\dagger}$ and self-adjointness of the linear operators $\d_X^2 v$ and $\d_Y^2 v$ are used.
Hence, the composition of the differential operators (\ref{Lambda}), associated with $N_1, N_2 \in \mR^{n\x n}$, is computed as
\begin{align}
\nonumber
    \Lambda_{N_1} (\Lambda_{N_2}(v))
     =&
    \frac{1}{2}
    \bra
        X N_1,
        \d_X
        \Lambda_{N_2}(v)
    \ket +
    \bra
        YN_1,
        \d_Y
        \Lambda_{N_2}(v)
    \ket\\
\nonumber
     =&
    \frac{1}{4}
    \bra
        X N_1 N_2, \d_X v
    \ket+
    \frac{1}{4}
    \bra
        N_2^{\rT} Y N_1 + N_1^{\rT} Y N_2
        + YN_1N_2 + N_2^{\rT}N_1^{\rT}Y,
        \d_Y v
    \ket\\
\nonumber
    &+
    \frac{1}{2}
    \bra
        XN_1,
        \d_X^2 v(XN_2)
    \ket
    +
    \frac{1}{2}
    \bra
        XN_1,
        \d_Y\d_X v(YN_2)
    \ket\\
\label{LamLam}
    &+
    \frac{1}{2}
    \bra
        YN_1,
        \d_X\d_Y v(XN_2)
    \ket+
    \frac{1}{4}
    \bra
        YN_1 + N_1^{\rT}Y,
        \d_Y^2 v(YN_2+N_2^{\rT}Y)
    \ket.
\end{align}
Note that the part of the right-hand side of (\ref{LamLam}), which involves the second-order derivatives  of $v$, is invariant under the transposition $(N_1,N_2)\mapsto (N_2,N_1)$, and so also is the matrix $N_2^{\rT} Y N_1 + N_1^{\rT} Y N_2$. Therefore, the commutator of $\Lambda_{N_1}$ and $\Lambda_{N_2}$ takes the form
\begin{eqnarray}
\nonumber
    [\Lambda_{N_1},\Lambda_{N_2}](v)
    & = &
    \frac{1}{4}
    \bra
        X [N_1, N_2],
        \d_X v
    \ket
    +
    \frac{1}{2}
    \bra
        Y [N_1, N_2],
        \d_Y v
    \ket\\
\nonumber
    &=&
    \frac{1}{2}
    \bra
        [N_1, N_2],
        M(v)
    \ket
    =
    \Lambda_{[N_1,N_2]/2}(v),
\end{eqnarray}
which holds for twice continuously Frechet differentiable functions $v$, thus establishing (\ref{LambdaLambda}).
\end{proof}

In view of the Frobenius theorem mentioned above, Lemma~\ref{lem:Frob} implies the local complete integrability of the PDE
\begin{equation}
\label{MV0}
        M(V)
        =0,
\end{equation}
whose left-hand side is given by (\ref{M}). The solutions of (\ref{MV0}) are also solutions of the PDE (\ref{MVJ}), but not visa versa. Since (\ref{MV0}) is a system of $n^2$ independent scalar PDEs, the integrability suggests that it has a  $n(n+1)/2$-dimensional integral manifold which can be represented using an $\mS_n$-valued map. Moreover, Lemma~\ref{lem:Frob}  also establishes the involutivity for the PDE (\ref{MVJ}), with its $n^2$-dimensional integral manifold representable by an $\mR^{n\x n}$-valued map, or a pair of maps with values in $\mS_n$ and $\mA_n$, where $\mA_n:= \mS_n^{\bot}$ is the subspace of real antisymmetric matrices of order $n$, which is the orthogonal complement of the subspace $\mS_n$ in the sense of the Frobenius inner product in $\mR^{n\x n}$. Indeed,
$$
    [Z_1J_0, Z_2J_0]
    =
    (Z_1J_0Z_2-Z_2J_0Z_1)J_0
$$
for any $Z_1, Z_2 \in \mS_n$, in accordance with the fact that the commutator of Hamiltonian matrices is also a Hamiltonian matrix.   Therefore, by applying (\ref{LambdaLambda}) to the operators (\ref{LZ}), it follows that the involutivity condition (\ref{LLL}) holds with $Z_3 := (Z_2J_0Z_1-Z_1J_0Z_2)/2\in \mS_n$, which implies the local complete integrability for the PDE (\ref{MVJ}).

\subsection{Two particular solutions}\label{sec:two}


\begin{lemma}
\label{lem:partsol1}
Suppose $f: \mS_n \to \mR$ is a Frechet differentiable function. Then the function
\begin{equation}
\label{VW}
    V(X,Y)
    :=
    f(XY^{-1} X^{\rT})
\end{equation}
defined for $X \in \mR^{n\x n}$ and $Y \in \mS_n$ with $\det Y \ne 0$, satisfies the PDE (\ref{MV0}).
Moreover, (\ref{VW}) describes the general smooth solution of (\ref{MV0}) on every connected component of the set $\{(X,Y)\in \mR^{n\x n}\x \mS_n:\, \det(XY)\ne 0\}$.
\end{lemma}
\begin{proof}
With the matrices $X\in \mR^{n\x n}$ and $Y\in \mS_n$,  where $\det Y\ne 0$, we associate the matrices
\begin{equation}
\label{UW}
    U
    :=
    XY^{-1},
    \quad
    W
    :=
    XY^{-1}X^{\rT}.
\end{equation}
The Frechet derivatives of $W$ with respect to $X$ and  $Y$ are expressed in terms of special linear operators of grade one (see Appendix~\ref{sec:class}) and the matrix transpose operator $\cT$ as
\begin{equation}
\label{dW}
    \d_X
    W
    =
    \[[[
        I,
        U^{\rT}
    \]]]
    +
    \[[[
        U,
        I
    \]]]
    \cT,
    \qquad
    \d_Y
    W
    =
    -
    \[[[
        U,
        U^{\rT}
    \]]],
\end{equation}
where the composition $M\circ N$ of linear operators $M$ and $N$ is written briefly as $MN$.
Indeed, the first variation of the matrix-valued map $(X,Y)\mapsto W$ in (\ref{UW}) is  computed as
\begin{align}
\nonumber
    \delta W
     &=
    (\delta X) Y^{-1}X^{\rT}
    +
    X Y^{-1}(\delta X)^{\rT}
    -
    XY^{-1}(\delta Y)Y^{-1}X^{\rT} \\
\label{deltaW}
     &=
    (\delta X) U^{\rT} + U \delta X^{\rT} - U(\delta Y)U^{\rT},
\end{align}
which implies (\ref{dW}).
The Frechet  derivatives of the composite function $V= f\circ W$ from (\ref{VW}) are
\begin{equation}
\label{dXV}
    \!\d_X\! V
    \!\!=\!
    (\d_X\! W)^{\dagger}(f')=
    (
    \[[[
        I,
        U
    \]]]
    \!+\!
    \cT
    \[[[
        U^{\rT},
        I
    \]]]    )(f')
    \!=\!
    2f'U,\!\!\!\!\!
\end{equation}
where $f'$ is the $\mS_n$-valued Frechet derivative of the function $f$, and the relations $\[[[\alpha, \beta\]]]^{\dagger} = \[[[\alpha^{\rT}, \beta^{\rT}\]]]$ and $\cT^{\dagger} = \cT$ are used, with $(\cdot)^{\dagger}$ the adjoint with respect to the Frobenius inner product of matrices. By a similar reasoning,
\begin{equation}
\label{dYV}
    \d_Y V
    =
    (\d_Y W)^{\dagger}(f')
    =
    -
    U^{\rT}
    f'
    U.
\end{equation}
Substitution of (\ref{dXV}) and (\ref{dYV}) into the left-hand side of (\ref{MV0}) yields
\begin{equation}
\label{myPDE4}
        M(V)
        =
        (X^{\rT} - YU^{\rT}) f'U=0,
\end{equation}
so that the function $V$ given by (\ref{VW}) indeed satisfies the PDE. We will now show that   (\ref{VW}) is, in fact, the general smooth solution of the PDE (\ref{MV0}) under the additional condition $\det X \ne 0 $, in which case both $U$ and $W$ in (\ref{UW}) are nonsingular. To this end, using the ideas of the method of characteristics for conventional PDEs \cite{E_1998,V_1971},
we will prove that any smooth function $V$ satisfying the PDE (\ref{MV0}) is constant on every connected component of the preimage
\begin{eqnarray}
\nonumber
    W^{-1}(S)
    & :=&
    \{
        (X,Y)\in \mR^{n\x n}\x \mS_n:\\
\label{invW}
        & &
        \det Y \ne 0, \,
        XY^{-1}X^{\rT}
        = S
    \}
\end{eqnarray}
of any given nonsingular matrix $S \in \mS_n$ under the map $(X,Y)\mapsto W$, with $W^{-1}$ the functional inverse. Indeed, let $[0,1] \ni s\mapsto (X,Y)\in W^{-1}(S)$ be a smooth curve lying in this set. By differentiating  the map $W$ along such a curve and using (\ref{deltaW}), it follows that
$
    0
    =
    \dot{W}
    =
    \dot{X}U^{\rT} + U\dot{X}^{\rT} - U\dot{Y}U^{\rT}
$,
which, in view of $\det U \ne 0$, allows $\dot{Y}$ to be expressed in terms of $\dot{X}$ as
\begin{equation}
\label{Ydot}
    \dot{Y} = U^{-1}\dot{X} + \dot{X}^{\rT}U^{-\rT},
\end{equation}
where $\dot{(\,)}:= \d_s$. Hence, differentiation of $V$ as a composite function along the curve yields
\begin{eqnarray}
\nonumber
    \dot{V}
    & = &
    \bra
        \d_X V,
        \dot{X}
    \ket
    +
    \bra
        \d_Y V,
        \dot{Y}
    \ket \\
\nonumber
& = &
    \bra
        \d_X V,
        \dot{X}
    \ket
    +
    \bra
        \d_Y V,
        U^{-1}\dot{X} + \dot{X}^{\rT}U^{-\rT}
    \ket \\
\label{Vdot}
    & = &
    \bra
        \d_X V
        +
        2 U^{-\rT} \d_Y V,
        \dot{X}
    \ket,
\end{eqnarray}
where use is made of (\ref{Ydot}) and
the symmetry of the matrix $\d_YV$. Since the PDE (\ref{MV0}) implies that
$
            \d_X V
        +
        2 U^{-\rT} \d_Y V
        =
        2 X^{-\rT}
        M(V)
        =0
$,
then (\ref{Vdot}) yields $\dot{V} = 0$. Hence, every smooth solution $V$ of (\ref{MV0}) is constant over any connected component of the set $W^{-1}(S)$ from (\ref{invW}). Indeed, existence of two distinct points,  which are connected by a smooth curve in $W^{-1}(S)$ and such that $V$ takes different values at these endpoints, would contradict the constancy of $V$ along any such curve established above. Thus, $V(X,Y)$ can only depend on $X$ and $Y$ through their special combination $XY^{-1}X^{\rT}$, and the ODE (\ref{Ydot}) generates characteristic curves on which  smooth solutions of the PDE are constant.
\end{proof}

Note that since (\ref{VW}) involves the matrix inverse $Y^{-1}$, the explicit representation of the solution would be hard to guess by treating (\ref{MV0}) as a system of scalar PDEs.

\begin{lemma}
\label{lem:partsol2}
Suppose $f: \mA_n\to \mR$ is a Frechet differentiable function. Then the function $V: \mR^{n\x n}\to \mR$ defined by
\begin{equation}
\label{VX}
    V(X)
    :=
    f(XJ_0 X^{\rT}),
\end{equation}
with $J_0$ the canonical antisymmetric matrix of order $n$ from (\ref{J0}), satisfies the PDE
\begin{equation}
\label{VXPDE}
    \bS(X^{\rT} \d_X V J_0)
    =
    0.
\end{equation}
Moreover, (\ref{VX}) describes the general solution of (\ref{VXPDE}) among Frechet differentiable functions of $X$ over any connected component of the set $\det X\ne 0$.
\end{lemma}
\begin{proof}
Since
$$
    \d_X (XJ_0 X^{\rT})
    =
    \[[[
        I,
        J_0 X^{\rT}
    \]]]
    +
    \[[[
        XJ_0,
        I
    \]]]
    \cT,
$$
then differentiation of (\ref{VX}) as a composite function of $X$ yields
$$    \d_X V
    =
    -f'XJ_0
    -(J_0 X^{\rT}f')^{\rT}
    =
    -2f'XJ_0,
$$
where we have also used the antisymmetry of the matrix $f'$. Hence,
$
    X^{\rT}\d_X V J_0 = 2X^{\rT} f' X
$
is antisymmetric whence (\ref{VXPDE}) follows. Now, to prove the converse, let $[0,1]\ni s \mapsto X \in \mR^{n\x n}$ be an arbitrary smooth curve in the set $\{X \in \mR^{n \x n}:\, XJ_0 X^{\rT}= \Omega\}$, where $\Omega \in \mA_n$ is a given nonsingular antisymmetric matrix. By differentiating $XJ_0 X^{\rT}$ along such a curve, it follows that
\begin{equation}
\label{XJX}
    \dot{X}J_0 X^{\rT} + X J_0 \dot{X}^{\rT}= 0.
\end{equation}
Since $\det X \ne 0$, then the left multiplication of (\ref{XJX}) by $X^{-1}$ and right multiplication by $X^{-\rT}$ yields
\begin{equation}
\label{XX}
    \bA(X^{-1}\dot{X} J_0)= 0,
\end{equation}
with $\bA$ the \textit{antisymmetrizer} defined by
the orthogonal projection onto the subspace $\mA_n$ of real antisymmetric matrices of order $n$ as
$$    \bA(N)
    :=
    N- \bS(N)
    =
    (N-N^{\rT})/2.
$$
In view of (\ref{bH}), the relation (\ref{XX}) is equivalent to the matrix $X^{-1}\dot{X}$ being Hamiltonian. Now, if  $V$ is an arbitrary smooth solution of the PDE (\ref{VXPDE}), then its derivative along the curve is
\begin{align}
\nonumber
    \dot{V}
    &=
    \bra
        \d_X V,
        \dot{X}
    \ket
    =
    \bra
        X^{\rT}\d_X V J_0,
        X^{-1}\dot{X} J_0
    \ket \\
\nonumber
     &=
    \bra
        \bS(X^{\rT}\d_X V J_0),
        \bS(X^{-1}\dot{X} J_0)
    \ket\\
\label{VXdot}
    &+
    \bra
        \bA(X^{\rT}\d_X V J_0),
        \bA(X^{-1}\dot{X} J_0)
    \ket
%
    =0,
\end{align}
where the Frobenius inner product is partitioned according to the orthogonal decomposition $\mR^{n \x n} = \mS_n \oplus \mA_n$. In view of (\ref{VXdot}), the solution $V$ is constant over any connected component of the set where $XJ_0 X^{\rT}$ is a given nonsingular matrix, thus implying the representation
(\ref{VX}).
\end{proof}
\subsection{General solution}\label{sec:gen}

The following theorem shows that the particular solutions of the PDE (\ref{MVJ}), obtained in the previous section, can be ``assembled'' into the general solution.

\begin{theorem}
\label{th:gensol}
Suppose $f: \mR^{n\x n}\to \mR$ is a Frechet differentiable function. Then the function
\begin{equation}
\label{VXY}
    V(X,Y)
    :=
    f(X(Y^{-1} + J_0)X^{\rT}),
\end{equation}
defined for $X\in \mR^{n\x n}$ and $Y \in \mS_n$, with $\det Y \ne 0$, satisfies the PDE (\ref{MVJ}). Moreover, (\ref{VXY}) is a general smooth solution of the PDE over any connected component of the set $\{(X,Y)\in \mR^{n\x n}\x \mS_n:\, \det(XY)\ne 0\}$.
\end{theorem}
\begin{proof}
Since the function (\ref{VXY}) can be represented as $V = g(XY^{-1}X^{\rT},XJ_0 X^{\rT})$, where  $g:\mS_n \x \mA_n\to \mR$ is another Frechet differentiable function given by
\begin{equation}
\label{fg}
    g(\sigma, \omega)
    :=
    f(\sigma+\omega),
\end{equation}
then the first claim of the theorem follows from the corresponding statements of Lemmas~\ref{lem:partsol1} and \ref{lem:partsol2}. The fulfillment of the PDE (\ref{MVJ}) for the function (\ref{VXY}) can also be verified directly using its partial Frechet derivatives
\begin{eqnarray}
\label{dXVXY}
    \d_X V
    & = &
    2(\bS(f')XY^{-1} - \bA (f')XJ_0),\\
\label{dYVXY}
    \d_Y V
    & = &
    -Y^{-1} X^{\rT} \bS(f') XY^{-1},
\end{eqnarray}
which follow from the relations $\d_{\sigma} g = \bS(f')$ and $\d_{\omega}g= \bA (f')$ for the function $g$ in (\ref{fg}). Now, to prove that (\ref{VXY}) is, in fact, the general solution of the PDE over any connected  component of the set $\det (XY)\ne 0$, we employ the transformation $(X,Y)\mapsto (X,W)$, with  $W$ given by (\ref{UW}). This is a diffeomorphism since, for any nonsingular $X$, the matrix $Y$ is uniquely and smoothly recovered from $W$ as $Y = X^{\rT} W^{-1} X$. The action of the operator $M$ from (\ref{M}) on the function  $h(X,W):= V(X,Y)$ written in the new independent variables $X$ and $W$ takes the form
\begin{eqnarray}
\nonumber
    M(V)
    & = &
    \frac{1}{2}
    X^{\rT} (\d_X h + (\d_X W)^{\dagger}(\d_W h))\\
\label{Mh}
    & &+
    X^{\rT} W^{-1} X (\d_Y W)^{\dagger}(\d_W h)
     =
    \frac{1}{2} X^{\rT}\d_X h,
\end{eqnarray}
where the terms containing $\d_W h$ cancel each other due to the structure of the operators $\d_X W$ and $\d_Y W$ from (\ref{dW}) employed in the proof of Lemma~\ref{lem:partsol1}. Substitution of (\ref{Mh}) into the PDE (\ref{MVJ}) leads to the PDE $\bS(X^{\rT} \d_X h J_0) = 0$. By considering this last PDE for a fixed but otherwise arbitrary nonsingular $W \in \mS_n$, and applying Lemma~\ref{lem:partsol2}, it follows that its general solution  over any connected component of the set $\det X \ne 0$ is described by $h(X,W) = \varphi(W, XJ_0 X^{\rT})$, where $\varphi: \mS_n \x \mA_n\to \mR$ is a Frechet differentiable function. Since any such $\varphi$ can be identified with a Frechet differentiable function $f: \mR^{n \x n}\to \mR$ by $f(N) = \varphi(\bS(N), \bA (N))$, cf. (\ref{fg}),  this proves the second claim of Theorem~\ref{th:gensol}.
\end{proof}

Finally, Theorem~\ref{th:Vshape} is obtained by applying Theorem~\ref{th:gensol} to the minimum cost  function $V_t(P)$ for fixed but otherwise arbitrary $t$ and $P_{11}$, assuming its Frechet smoothness on the set where the blocks of the covariance matrix $P$ from (\ref{Pblocks}) satisfy $\det P_{12}\ne 0$ and $P_{22}\succ 0$.


\section{Special linear operators on matrices}\label{sec:class}

Following \cite{VP_2011}, we define, for any matrices $\alpha \in \mR^{s \x p}$ and  $\beta\in \mR^{q\x t}$,   a linear operator $\[[[\alpha, \beta\]]]: \mR^{p\x q}\to \mR^{s\x t}$ by
\begin{equation}
\label{cL}
    \[[[\alpha, \beta\]]](X)
    :=
    \alpha X \beta.
\end{equation}
The generalization of this construct to matrices $\alpha_1, \ldots, \alpha_r \in \mR^{s\x p}$ and $\beta_1, \ldots, \beta_r\in \mR^{q\x t}$, with $r$ an arbitrary positive integer, leads to a linear operator
\begin{equation}
\label{cLsum}
    \[[[
        \alpha_1,\beta_1
        \mid
        \ldots
        \mid
        \alpha_r, \beta_r
    \]]]
    :=
    \sum_{k=1}^{r}
    \[[[
        \alpha_k,  \beta_k
    \]]],
\end{equation}
where the matrix pairs are separated by ``$\mid$''s.
Of particular importance are
self-adjoint linear operators on the Hilbert space
$\mR^{p\x q}$  of the form (\ref{cLsum}) where
$\alpha_1, \ldots, \alpha_r\in \mR^{p\x p}$ and $\beta_1, \ldots, \beta_r\in \mR^{q\x
q}$ are such that for any $k=1, \ldots, r$, the matrices $\alpha_k$ and $\beta_k$ are either both symmetric or both antisymmetric. Such an operator  (\ref{cLsum}) is referred to as a \textit{self-adjoint operator of grade} $r$, with the self-adjointness understood in the sense of the Frobenius inner product on $\mR^{p\x q}$, so that  $\[[[ \alpha, \beta\]]]^{\dagger} = \[[[\alpha^{\rT}, \beta^{\rT}\]]]$.
%
\begin{lemma}
\label{lem:spec}\cite{VP_2011}
If $\alpha \in \mR^{p\x p}$ and $\beta\in \mR^{q\x q} $ are both antisymmetric, then the spectrum of $\[[[\alpha, \beta\]]]$ is symmetric about the origin. If $\alpha$ and $\beta $ are both symmetric and positive (semi-) definite, then $\[[[\alpha, \beta\]]]$ is positive (semi-) definite, respectively.
\end{lemma}

Whilst the operator (\ref{cL})  with nonsingular $\alpha$ and
$\beta$ is straightforwardly
invertible, so that $\[[[\alpha,  \beta\]]]^{-1} = \[[[\alpha^{-1},\beta^{-1}\]]]$, the
inverse of the operator from (\ref{cLsum})
with $r>1$, in general,
can only be computed
using the vectorization of matrices~\cite{M_1988,SIG_1998} as
$$        \[[[
        \alpha_1,\beta_1
        \mid
        \ldots
        \mid
        \alpha_r, \beta_r
    \]]]^{-1}(Y)
    =
    \col^{-1}
    (
        \gamma^{-1}
    \col(Y)
    ),
$$
provided
$    \gamma
    :=
    \sum_{k=1}^{r}
    \beta_k^{\rT}\ox \alpha_k
$ is nonsingular.
Here, $\col(Y)$ is the vector obtained by writing the columns of a matrix $Y$  one underneath
the other.

\end{document}